\documentclass[10pt, conference]{IEEEtran}

\usepackage{amsmath,amssymb,amsthm,graphicx,enumerate,verbatim,xcolor}
\usepackage{tikz}
\usetikzlibrary{arrows,positioning,fit,backgrounds}
\usetikzlibrary{calc}

\newcommand{\eps}{\varepsilon}

\newcommand{\Ebb}{\mathbb{E}}
\newcommand{\Nbb}{\mathbb{N}}
\newcommand{\Pbb}{\mathbb{P}}

\newcommand{\Acal}{\mathcal{A}}
\newcommand{\Bcal}{\mathcal{B}}

\newcommand{\Mcal}{\mathcal{M}}
\newcommand{\Kcal}{\mathcal{K}}
\newcommand{\Pcal}{\mathcal{P}}

\newcommand{\Xcal}{\mathcal{X}}
\newcommand{\Ycal}{\mathcal{Y}}
\newcommand{\Zcal}{\mathcal{Z}}

\newtheorem{thm}{Theorem}
\newtheorem{cor}{Corollary}
\newtheorem{lemma}{Lemma}
\newtheorem{defn}{Definition}

\title{Secrecy Is Cheap if the Adversary Must Reconstruct}
\author{
\authorblockN{Curt Schieler, Paul Cuff}
\authorblockA{Dept. of Electrical Engineering,\\
Princeton University,
Princeton, NJ 08544.\\
E-mail: \{schieler, cuff\}@princeton.edu }
}
\begin{document}
\maketitle
\begin{abstract}
A secret key can be used to conceal information from an eavesdropper during communication, as in Shannon's cipher system.  Most theoretical guarantees of secrecy require the secret key space to grow exponentially with the length of communication.  Here we show that when an eavesdropper attempts to reconstruct an information sequence, as posed in the literature by Yamamoto, very little secret key is required to effect unconditionally maximal distortion; specifically, we only need the secret key space to increase unboundedly, growing arbitrarily slowly with the blocklength.  As a corollary, even with a secret key of constant size we can still cause the adversary arbitrarily close to maximal distortion, regardless of the length of the information sequence.
\end{abstract}

\section{Introduction}

In this work, we consider the Shannon cipher system, first investigated in \cite{Shannon1949}. The cipher system is a communication system with the addition of secret key that the legitimate parties share and use to encrypt messages. A classic result by Shannon in \cite{Shannon1949} states that to achieve perfect secrecy, the size of the secret key space must be at least the size of the message space.  As in \cite{Shannon1949}, we consider the secrecy resource to be shared secret key, but we relax the requirement of perfect secrecy and instead look at the minimum distortion that an adversary attains when attempting to reproduce the source sequence. The joint goal of Alice (transmitter) and Bob (receiver) is to communicate a source sequence almost losslessly while maximizing the adversary's minimum attainable distortion. In contrast to equivocation, a max-min distortion measure provides guarantees about the way in which any adversary could use his knowledge; equivocation does not give much insight into the structure of the knowledge or how the knowledge can be used.

This measure of security was investigated by Yamamoto in the general case where distortion is allowed at the legitimate receiver. In \cite{Yamamoto1997}, Yamamoto established upper and lower bounds on the tradeoff between the rate of secret key and the adversary's distortion.

In this paper, we solve the problem studied in \cite{Yamamoto1997}, in the case that almost lossless communication is required. We show that any positive rate of secret key suffices for Alice and Bob to cause the adversary unconditionally maximal distortion (i.e., the distortion incurred by only knowing the source distribution and nothing else). A positive rate of secret key $R_0$ means the number of secret keys is exponential in the blocklength $n$, because there are $2^{nR_0}$ secret keys available. However, if the secret key space is merely growing unboundedly with $n$, we show that the adversary still suffers maximal distortion. We also show that a constant amount of secret key can yield nontrivial distortion at the adversary.

\section{Problem Statement}

The system under consideration, shown in Figure \ref{setup_fig}, operates on blocks of length $n$.
\begin{figure}
  \begin{tikzpicture}
 [node distance=1cm,minimum height=7mm,minimum width=14mm,arw/.style={->,>=stealth'}]
  \node[coordinate] (source) {};
  \node[rectangle,draw,rounded corners] (alice) [right =of source] {Alice};
  \node[coordinate] (dummy) [right =1cm of alice] {};
  \node[rectangle,draw,rounded corners] (bob) [right =1cm of dummy] {Bob};
  \node[coordinate] (xhat) [right =of bob] {};
  \node[rectangle,minimum width=7mm] (key) [above =of dummy] {$K$};
  \node[rectangle,draw,rounded corners] (eve) [below =5mm of bob] {Eve};
  \node[coordinate] (zn) [right =of eve] {};

  \draw [arw] (source) to node[midway,above,yshift=-1mm]{$X^n$} (alice);
  \draw [arw] (alice) to node[midway,above,yshift=-1mm]{$M$} (bob);
  \draw [arw] (bob) to node[midway,above,yshift=-.5mm]{$\hat{X}^n$}(xhat);
  \draw [arw] (key) to [out=180,in=90] (alice);
  \draw [arw] (key) to [out=0,in=90] (bob);
  \draw [arw,rounded corners] (dummy) |- (eve);
  \draw [arw] (eve) to node[midway,above,yshift=-1mm]{$Z^n$} (zn);
 \end{tikzpicture}
 \caption{Alice and Bob share secret key $K$, which Alice uses along with her observation of $X^n$ to encode a message $M$. Secrecy is measured by the minimum distortion Eve can attain.}
\label{setup_fig}
\end{figure}
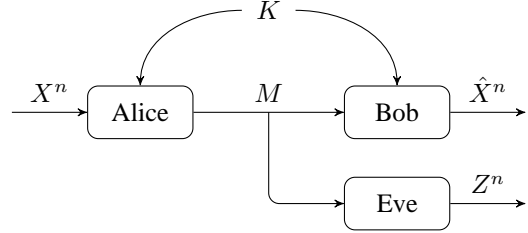
Alice is given an i.i.d. source sequence $X^n=(X_1,\ldots,X_n)$ consisting of symbols drawn from a finite alphabet $\Xcal$ according to $P_X$. Alice and Bob share secret key in the form of a uniform random variable $K$ taking values in an alphabet $\Kcal$. Eve knows the source distribution and the operations of Alice and Bob, but does not have access to the secret key. At the transmitter, Alice sends $M \in \Mcal$ based on the source sequence $X^n$ and secret key $K$; at the other end, Bob observes $M$ and $K$ and produces a sequence $\hat{X}^n$. Eve produces a sequence $Z^n$ from $M$ and her knowledge of the source distribution and system  operations (encoder and decoder).
\begin{defn}
 Let $k:\Nbb\rightarrow \Nbb$. An $(n,k(n),R)$ code consists of an encoder $f$ and a decoder $g$:
 \begin{IEEEeqnarray*}{l}
  f:\Xcal^n\times \Kcal \rightarrow \Mcal\\
  g: \Mcal \times \Kcal \rightarrow \Xcal^n,
 \end{IEEEeqnarray*}
where the size of the message set is $|\Mcal|=2^{nR}$, and the number of secret keys available is $|\Kcal|=k(n)$.
\end{defn}

We measure secrecy by the distortion between the source sequence and an adversary's estimate of the source sequence. Given a per-letter distortion measure $d:\Xcal\times\Zcal\rightarrow [0,\infty)$, we define the distortion between two sequences as the average of the per-letter distortions:
\begin{equation*}
 d^n(x^n,z^n) = \frac1n \sum_{i=1}^n d(x_i,z_i).
\end{equation*}
Without loss of generality, we assume that for all $x\in\Xcal$, there exists a $z\in\Zcal$ such that $d(x,z)=0$.

For a given amount of secret key, we are interested in the rate of communication and the distortion incurred by the cleverest adversary.
\begin{defn}
For a given sequence $k(n)$ and measure of distortion $d(x,z)$, we say that the pair $(R,D)$ is achievable if there exists a sequence of $(n,k(n),R)$ codes such that
\begin{equation}
 \label{errdefn}\lim_{n\rightarrow\infty}\Pbb[X^n\neq\hat{X}^n]=0
\end{equation}
and
\begin{equation}
 \label{paydefn}\liminf_{n\rightarrow\infty}\min_{z^n(m)}\Ebb\left[d^n(X^n,z^n(M))\right]\geq D.
\end{equation}
\end{defn}
The requirement in (\ref{errdefn}) is that the probability of communication error between Alice and Bob vanishes. In (\ref{paydefn}), the minimum is taken over all functions $z^n:\Mcal\rightarrow\Zcal^n$, i.e., all possible strategies that Eve can employ. Although not explicit in the notation, it should be understood that Eve's strategy is a function of not only the message $M$, but also the source distribution $P_X$ and the $(n,k(n),R)$ code.

\section{Main Result}

The main result is the following theorem. The restriction on $R$, the communication rate, is the same as the classic result for source coding. Notice that $\min_z \Ebb[d(X,z)]$ is the distortion between $X^n$ and the constant sequence $(z^*,\ldots,z^*)$, where $z^*=\mbox{argmin}_z \Ebb[d(X,z)]$.
\begin{thm}
 \label{mainthm}
 Let $k(n)$ be an increasing, unbounded sequence. Then $(R,D)$ is achievable if and only if\,\footnote{For simplicity, we ignore the case $R=H(X)$}
\begin{IEEEeqnarray*}{l}
 R>H(X)\\
 D\leq \min_z \Ebb[d(X,z)].
\end{IEEEeqnarray*}
\end{thm}
If we wanted to consider \emph{rates} of secret key, we would set $k(n)=2^{nR_0}$ and define $(R,R_0,D)$ to be achievable if there exists a sequence of codes such that (\ref{errdefn}) and (\ref{paydefn}) hold. Then, by Theorem~\ref{mainthm}, we would have that $(R,R_0,D)$ is achievable if and only if
\begin{IEEEeqnarray*}{lCl}
 R>H(X) && R>H(X)\\
 R_0>0 &\qquad\quad \text{or}\qquad\quad& R_0=0\\
 D\leq \min_z \Ebb[d(X,z)] &&  D=0
\end{IEEEeqnarray*}
This is the solution to the lossless case of the problem posed in \cite{Yamamoto1997}. It should be noted that with the proper choice of auxilliary random variables, the converse bound on $R_0$ in \cite{Yamamoto1997} is actually the trivial bound, $R_0\geq 0$.

With Theorem~\ref{mainthm} in hand, we are able to say something about the usefulness of a finite amount of secret key. The following corollary asserts that the cleverest adversary suffers close to maximal distortion even if the number of secret keys stays constant as blocklength increases.
\begin{cor}
 \label{finite_cor}
 Fix $P_X$ and $d(x,z)$, and denote $D_{\max}=\min_z \Ebb[d(X,z)]$. For all $D<D_{\max}$ and $R>H(X)$, there exists $k^*\in\Nbb$ such that $(R,D)$ is achievable under $k(n)=k^*$.
\end{cor}
\begin{proof}[Proof of Corollary~\ref{finite_cor}]
 Suppose the contrary. That is, assume there exists $D<D_{\max}$ or $R>H(X)$ such that for all $\tilde{k}\in\Nbb$, $(R,D)$ is not achievable under $k(n)=\tilde{k}$. If we denote the minimum attainable distortion for blocklength $n$ and $\tilde{k}$ secret keys by
 \begin{equation*}
  d_{n,\tilde{k}}=\min_{z^n(m)}\Ebb\left[d^n(X^n,z^n(M))\right],
 \end{equation*}
we are asserting that for all $(n,\tilde{k},R)$ codes, either
\begin{IEEEeqnarray}{rCl}
 \label{corstep1} &\quad& \limsup_{n\rightarrow\infty}\Pbb[X^n\neq\hat{X}^n]>0\\
 \label{corstep2} \text{or} &\quad& \liminf_{n\rightarrow\infty} d_{n,\tilde{k}} < D.
\end{IEEEeqnarray}
In particular, all $(n,\tilde{k},R)$ codes not satisfying (\ref{corstep1}) must satisfy (\ref{corstep2}), which implies that for all $\tilde{k}\in\Nbb$, the sequence $d_{n,\tilde{k}}$ is strictly less than $D$ infinitely often. To arrive at a contradiction, we will define an increasing unbounded sequence $\hat{k}(n)$ such that $d_{n,\hat{k}(n)}$ is strictly less than $D$ infinitely often. Since $D~<~D_{\max}$, such a $\hat{k}(n)$ will imply
\begin{equation*}
 \liminf_{n\rightarrow\infty} d_{n,\hat{k}(n)} < D_{\max},
\end{equation*}
 contradicting Theorem~\ref{mainthm} and completing the proof. To that end, first define the increasing sequence $\{N_\ell\}$ recursively by
\begin{IEEEeqnarray*}{l}
 N_0 = 0\\
 N_\ell = \min\{n>N_{\ell-1}:d_{n,\ell}<D\}.
\end{IEEEeqnarray*}
Then we define $\hat{k}(n)$ by
\begin{equation*}
 \hat{k}(n)=\ell \:\text{ if }\: N_{\ell-1} < n \leq N_\ell.
\end{equation*}
\end{proof}

The proof of Theorem~\ref{mainthm} is presented in the next section, but first we provide some intuition for why the result holds by briefly addressing some of the proof ideas. In designing a code, Alice and Bob can use the secret key $K$ to apply a one-time pad to part of the message so that the adversary effectively knows that the source sequence $X^n$ lies in a subset $B\subset\Xcal^n$, but is unsure which sequence the true one is. The number of sequences in $B$ is $|\Kcal|=k(n)$, the number of secret keys. Under such a scheme, the adversary's optimal strategy for minimizing distortion is to output the following symbol on the $i$th step:
\begin{equation}
 \label{opt_adv}
 z_i(B) = \text{argmin}_{z}\sum_{x^n\in B}\frac{p(x^n)}{\Pbb[X^n\in B]}\,d(x_i,z)
\end{equation}
Note that (\ref{opt_adv}) is the expected value of $d(X_i,z)$ conditioned on the event $\{X^n\in B\}$. Now, if each of the sequences in $A$ were equally likely to be the source sequence, (\ref{opt_adv}) becomes
\begin{IEEEeqnarray*}{rCl}
 z_i(B) &=& \text{argmin}_{z}\sum_{x^n\in B}\frac{1}{|B|}\,d(x_i,z)\\
 &=& \text{argmin}_{z}\sum_{x\in \Xcal} Q_i(x)\,d(x,z),
\end{IEEEeqnarray*}
where $Q_i(x)$ denotes the empirical distribution of the $i$th symbols of the sequences in $B$, i.e.,
\begin{equation*}
 Q_i(x) = \frac{1}{|B|}\sum_{x^n\in B} 1\{x_i=x\}.
\end{equation*}
If we could also guarantee that $Q_i(x)=P_X(x)$ for all $x\in\Xcal$, then (\ref{opt_adv}) would become
\begin{equation}
 \label{ideal}
 z_i(B) = \text{argmin}_{z}\,\Ebb[d(X,z)]
\end{equation}
In the light of this discussion, we want to design a codebook and an encryption scheme so that, roughly speaking,
\begin{equation}
 \label{approxunif}\frac{p(x^n)}{\Pbb[X^n\in B]}\approx \frac{1}{|B|}
\end{equation}
and
\begin{equation}
  \label{QapproxP} Q_i\approx P_X, \;i=1,\ldots,n.
\end{equation}

\begin{figure}
 \begin{equation*}
  B:\quad
  \begin{matrix}
   0 & 0 & 1 & 0 & 0 & 1 & 2 & 2\\
   0 & 2 & 0 & 1 & 1 & 2 & 0 & 0\\
   1 & 0 & 2 & 2 & 0 & 0 & 1 & 0\\
   2 & 1 & 0 & 0 & 2 & 0 & 0 & 1
  \end{matrix}
 \end{equation*}
 \caption{Consider $P_X=\{\frac12,\frac14,\frac14\}$, $\Xcal=\{0,1,2\}$, $n=4$, and $k(4)=8$. Suppose Eve knows that the source sequence $X^4$ is a column of $B$, but does not know which column. Since all the columns are equally likely and the empirical distribution of each row matches $P_X$, Eve's best strategy is to output $\mbox{argmin}_z \Ebb[d(X,z)]$ at each step (see (\ref{ideal})). For example, if the distortion measure were Hamming distance (i.e., $d(x,z)=1\{x\neq z\}$), then Eve would output $(0,0,0,0)$.}
 \label{bin_fig}
\end{figure}

Figure~\ref{bin_fig} gives an example of (\ref{approxunif}) and (\ref{QapproxP}). These ideas are borne out in the proof of Theorem~\ref{mainthm}, which we now turn to.

\section{Proof of Theorem~\ref{mainthm}}

 In preparation for the proof of achievability, we first define $\eps$-typicality for a distribution $P$ with finite support $\Xcal$:
\begin{equation*}
 T^n_\eps(P)=\{x^n\in\Xcal^n:\left|Q_{x^n}(x)-P(x)\right|<\eps,\forall x\in \Xcal\},
\end{equation*}
where $Q_{x^n}(x)=\frac1n \sum_i 1{\{x_i=x\}}$ is the empirical distribution, or ``type'', of $x^n$. Denote the set of types of sequences $x^n\in\Xcal^n$ by $\Pcal^n$, and let $\Pcal^n_\eps\subset \Pcal^n$ denote the set of types of those sequences $x^n\in\Xcal^n$ satisfying $x^n\in T^n_\eps(P_X)$. For $P\in\Pcal^n$, use $|P|$ to denote the number of sequences of type $P$. Finally, define the variational distance between distributions $P$ and $Q$ by
\begin{equation*}
 \lVert P-Q\rVert_{}=\sup_A|P(A)-Q(A)|.
\end{equation*}

We will need a few lemmas. The first three lemmas will aid us in asserting (\ref{QapproxP}).
\begin{lemma}
 \label{ColRowDistr}
Let $P\in\Pcal^n$. Form a matrix whose columns are the sequences with type $P$, with the columns arranged in any order. Then each of the rows of the matrix also has type $P$.
\end{lemma}
\begin{proof}
 Any permutation applied to the rows of the matrix simply permutes the columns. Therefore all the rows have identical type. Since the matrix as a whole has type $P$, each of the rows must be of type $P$ as well.
\end{proof}
\begin{lemma}[see \cite{Diaconis--Freedman1980}]
 \label{SampRepl}
 Suppose an urn $U$ contains $n$ balls, each marked by an element of the set $S$, whose cardinality $c$ is finite. Let $H$ be the distribution of $k$ draws made at random without replacement from $U$, and $M$ be the distribution of $k$ draws made at random with replacement. Thus, $H$ and $M$ are two distributions on $S^k$. Then
 $$\lVert H-M\rVert_{TV}\leq \frac{ck}{n}.$$
\end{lemma}
Thus, sampling without replacement is close in variational distance to sampling with replacement (i.e, i.i.d.) provided the sample size is small enough and the number of balls is large enough. The rate at which the distance vanishes is important to our problem. The next lemma is a lower bound on the size of a type class.
\begin{lemma}[see \cite{Csiszar--KornerBook}]
\label{TypeSize}
For $P\in\Pcal^n$,
$$|P|\geq (n+1)^{-|\Xcal|}|\Xcal|^{nH(P)}$$
\end{lemma}
The final lemma concerns sufficient statistics in the context of our measure of secrecy.
\begin{lemma}
\label{suffstat}
Let $X,Y$, and $Z$ be random variables that form a markov chain $X-Y-Z$ and let $g$ be a function on $\Acal \times \Zcal$. Define two sets of functions, $F=\{f:\Xcal\times\Ycal\rightarrow \Acal\}$ and $F'=\{f:\Ycal\rightarrow\Acal\}$. Then
\begin{equation*}
 \min_{f\in F}\Ebb[g(f(X,Y),Z)]=\min_{f\in F'}\Ebb[g(f(Y),Z)].
\end{equation*}
\end{lemma}
\begin{proof}[Proof of Lemma~\ref{suffstat}]
($\leq$) follows from $F'\subset F$. As for $(\geq)$, we have
 \begin{IEEEeqnarray*}{rCl}
  \min_{f\in F}\Ebb[g(f(X,Y),Z)] &=& \sum_{x,y}p(x,y)\sum_zp(z|y)g(f^*(x,y),z)\\
  &=& \sum_{x,y} p(x,y)h(x,y)\\
  &=& \sum_y p(y)\Ebb[h(X,Y)|Y=y]
 \end{IEEEeqnarray*}
There exists $x^*(y)$ such that $$h(x^*(y),y)\leq \Ebb[h(X,Y)|Y=y],$$ so we define $f\in F'$ by $f(y)=f^*(x^*(y),y)$. Then
\begin{IEEEeqnarray*}{rCl}
 \sum_y p(y)\Ebb[h(X,Y)|Y=y] &\geq& \sum_yp(y)h(x^*(y),y)\\
 &=& \sum_yp(y)\sum_zp(z|y)g(f(y),z)\\
 &=& \Ebb[g(f(Y),Z)]\\
 &\geq& \min_{f\in F'}\Ebb[g(f(Y),Z)] 
\end{IEEEeqnarray*}
\end{proof}
\bigskip
Now we begin the proof of Theorem~\ref{mainthm}.
\begin{proof}[Proof of Theorem~\ref{mainthm}]
 The proof of the converse is straightforward: the converse for lossless source coding gives us $R>H(X)$, and Eve can always produce the constant sequence $(z^*,\ldots,z^*)$ so that her distortion never exceeds $\min_z\Ebb[d(X,z)]$.

To begin the proof of achievability, fix $P_X$, $d(x,z)$, and an increasing, unbounded sequence $k(n)$. Let $\eps>0$ and $R~>~H(X)$. We will show that there exists a codebook of $2^{nR}$ sequences and an encryption scheme such that
\begin{equation}
 \label{err}
 \Pbb[X^n\neq\hat{X}^n]<\eps
\end{equation}
 and
\begin{equation}
 \label{pay}
 \min_{z^n(m)}\Ebb\left[d^n(X^n,z^n(M))\right] > \min_z \Ebb[d(X,z)]-\delta(\eps)
\end{equation}
for sufficiently large $n$, where $\delta(\eps)\rightarrow 0$ as $\eps\rightarrow0$.

Our codebook, the set of sequences that Alice encodes uniquely, consists of the $\eps$-typical sequences; thus, (\ref{err}) is satisfied by the law of large numbers. For blocklength $n$, we want to consider a partition of the set of typical sequences into equally sized subsets (or ``bins'') of length $k(n)$. A partition will let us encode the message in two parts: in the first part, we will reveal the identity of the bin that contains the source sequence, and in the second part we will encrypt the location within the bin by using the secret key to apply a one-time pad. Effectively, the second part of the message will be useless to Eve. We will denote the set of bins by $\Bcal$, so that each element of $\Bcal$ is a bin of $k(n)$ sequences.

For a given partition of the typical sequences, the encoder operates as follows. If $X^n$ is typical and is the $L$th sequence in bin $J$, then transmit the pair $(J,L\oplus K)$, where $K$ is the secret key and $\oplus$ is addition modulo $k(n)$. If $X^n$ is not typical, transmit a random message.

In addition to requiring equal-sized bins, we further restrict our attention to partitions in which each bin only contains sequences of the same type\footnote{More precisely, we focus on partitions in which the number of bins in violation is polynomial in $n$. The set of such partitions is nonempty since the total number of types is polynomial in $n$ (see \cite{Csiszar--KornerBook}). The forthcoming analysis is easily adjusted accordingly.}, and denote the set of bins of type $P$ by $\Bcal_P$; thus, $\Bcal=\bigcup_{P\in\Pcal_\eps^n} \Bcal_P$. This restriction addresses (\ref{approxunif}).

We claim that there exists a partition so that (\ref{pay}) is satisfied. To do this, we first select a partition uniformly at random and average the minimum attainable distortion over all partitions. We use $\Ebb_\pi$ to indicate that expectation is being taken with respect to a random partition. If (\ref{pay}) holds for the average, then it must hold for at least one partition. This use of the probabilistic method should be distinguished from ``random binning'' that is often used in information theory. In random binning, each sequence is assigned to a random bin; in particular, the bin sizes are random, whereas here they are of size $k(n)$.

Selecting a partition at random is the same as drawing typical sequences without replacement to fill equal-sized bins of uniform type. This is also equivalent to first fixing a partition $\Bcal$ that meets the criteria, then for each $P\in\Pcal_\eps^n$ randomly permuting the sequences in $\Bcal_P$, selecting the $|\Pcal_\eps^n|$ random permutations independently.

Denoting the left-hand side of (\ref{pay}) by $D(n)$, we first use Lemma~\ref{suffstat}, then restrict attention to typical sequences to get
\begin{IEEEeqnarray*}{rCl}
 \Ebb_\pi[D(n)] &=& \Ebb_\pi \left[\min_{z^n(j,l)}\Ebb\left[d^n(X^n,z^n(J,L\oplus K))\right]\right]\\
 &=&  \Ebb_\pi \left[\min_{z^n(j)}\Ebb\left[d^n(X^n,z^n(J))\right]\right]\\
 &\geq& \Ebb_\pi \bigg[\min_{z^n(j)}\sum_{x^n\in T^n_\eps(P_X)}p(x^n)d^n(x^n,z^n(J(x^n)))\bigg]
\end{IEEEeqnarray*}
Note that although $x^n$ is deterministic when inside the summation above, the bin $J(x^n)$ that it belongs to is random because we are considering a random partition. Summing over bins and moving the summation outside, we have
\begin{IEEEeqnarray*}{rCl}
 \Ebb_\pi[D(n)] &\geq& \Ebb_\pi\bigg[\min_{z^n(j)}\sum_{B\in \Bcal}\sum_{x^n\in B}p(x^n)d^n(x^n,z^n(J(x^n))\bigg]\\
 &=& \Ebb_\pi\bigg[\sum_{B\in \Bcal}\min_{z^n}\sum_{x^n\in B}p(x^n)d^n(x^n,z^n)\bigg]
\end{IEEEeqnarray*}
Next, we sum over types as well, and use the fact that all sequences of type $P$ have probability $$c_P=|\Xcal|^{n(H(P)+D(P||P_X))}$$ to get
\begin{IEEEeqnarray*}{rCl}
 \IEEEeqnarraymulticol{3}{l}{\Ebb_\pi[D(n)]}\\
 &\geq& \Ebb_\pi\bigg[\sum_{P\in \Pcal^n_\eps}\sum_{B\in \Bcal_P}\min_{z^n}\sum_{x^n\in B}p(x^n)d^n(x^n,z^n)\bigg]\\
 &=& \Ebb_\pi\bigg[\sum_{P\in \Pcal^n_\eps}\sum_{B\in \Bcal_P}c_P\min_{z^n}\sum_{x^n\in B}d^n(x^n,z^n)\bigg]
\end{IEEEeqnarray*}
Applying the separability of $d^n(x^n,z^n)$ and moving the expectation inside, we have
\begin{IEEEeqnarray*}{rCl}
 \IEEEeqnarraymulticol{3}{l}{\Ebb_\pi[D(n)]}\\
 &\geq& \Ebb_\pi\bigg[\frac1n\sum_{i=1}^n\sum_{P\in \Pcal^n_\eps}\sum_{B\in \Bcal_P}c_P\min_{z}\sum_{x^n\in B}d(x_i,z)\bigg]\\
 \yesnumber \label{expect1}&=& \frac1n\sum_{i=1}^n\sum_{P\in \Pcal^n_\eps}\sum_{B\in \Bcal_P}c_P\,\Ebb_\pi\bigg[\min_{z}\sum_{x^n\in B}d(x_i,z)\bigg]
\end{IEEEeqnarray*}
Keep in mind that the elements of $B$ are random codewords because the partition is random. 

Now we analyze the expectation in (\ref{expect1}). Viewing $\Bcal_P$ as a matrix with the constituent sequences forming the columns, we denote the $i$th row by the random sequence $(Y_1,\ldots,Y_{|P|})$. Furthermore, we let $(Y_1,\ldots,Y_{k(n)})$ denote the $i$th row of $B~\in~\Bcal_P$; this is acceptable because the forthcoming analysis is the same for each row of each bin. For ease of exposition, we now refer to $k(n)$ as simply $k$ with the dependence on $n$ understood. Thus, we have
\begin{equation*}
 \Ebb_\pi\bigg[\min_{z}\sum_{x^n\in B}d(x_i,z)\bigg] = k\cdot\Ebb_\pi\bigg[\min_{z}\sum_{x\in \Xcal}Q_{Y^k}(x)d(x,z)\bigg]
\end{equation*}
where $Q_{Y^k}$ is the type of $Y^k$. Denoting the event $\{Y^k~\in~T_\eps^k(P)\}$ by $A$, we have by the towering property of expectation that
\begin{IEEEeqnarray*}{rCl}
 \IEEEeqnarraymulticol{3}{l}{\Ebb_\pi\bigg[\min_{z}\sum_{x^n\in B}d(x_i,z)\bigg]}\\
 \yesnumber \label{expect2}&\geq& k\cdot\Pbb_\pi[A]\cdot\Ebb_\pi\bigg[\min_{z}\sum_{x\in \Xcal}Q_{Y^k}(x)d(x,z)\,\Big\vert\, A\bigg].
\end{IEEEeqnarray*}
Focusing attention on the conditional expectation in (\ref{expect2}), we use the definition of typicality and the triangle inequality to get
\begin{IEEEeqnarray*}{rCl}
 \IEEEeqnarraymulticol{3}{l}{\Ebb_\pi\bigg[\min_{z}\sum_{x\in \Xcal}Q_{Y^k}(x)d(x,z)\,\Big\vert\, A\bigg]}\\
 &\geq& \Ebb_\pi\bigg[\min_{z}\sum_{x\in \Xcal}(P_X(x)-2\eps)d(x,z)\,\Big\vert\, A\bigg] \\
 &=& \min_{z}\sum_{x\in \Xcal}(P_X(x)-2\eps)d(x,z)\\
 \yesnumber \label{step1} &\geq& \min_{z} \Ebb[d(X,z)]-\delta_1(\eps)
\end{IEEEeqnarray*}
where $\delta_1(\eps)=2\eps\min_z\sum_x d(x,z)$ goes to zero as $\eps\rightarrow0$ because the distortion measure $d$ is bounded. Now we bound $\Pbb[A]$ in (\ref{expect2}). We can assume that $k(n)\in o(|\Xcal|^{nH(P)})$ without loss of generality because Alice and Bob can simply ignore extra secret key. Invoking Lemmas~\ref{ColRowDistr}-\ref{TypeSize} to address (\ref{QapproxP}), we have
\begin{IEEEeqnarray}{rCl}
 \label{tv1}\Big\lVert P_{Y^k}-\prod_k P \Big\rVert &=& \Big\lVert P_{Y^k}-\prod_k Q_{Y^{|P|}} \Big\rVert\\
 \label{tv2}&\leq& \frac{|\Xcal|\cdot k(n)}{|P|}\\
 \label{tv3}&\leq& \frac{|\Xcal|\cdot k(n)}{(n+1)^{-|\Xcal|}|\Xcal|^{nH(P)}}\\
 \label{tv4}&\leq& \eps
\end{IEEEeqnarray}
for large enough $n$, where (\ref{tv1}) follows from Lemma~\ref{ColRowDistr}, (\ref{tv2}) follows from Lemma~\ref{SampRepl}, and  (\ref{tv3}) follows from Lemma~\ref{TypeSize}. By the definition of variational distance and the law of large numbers, (\ref{tv4}) gives
\begin{IEEEeqnarray*}{rCl}
 \Pbb_\pi[Y^k\in T_\eps^n(P)] &\geq& \Pbb_{\prod P}[Y^k\in T_\eps^n(P)]-\eps \\
 \yesnumber \label{step2} &\geq& 1-2\eps
\end{IEEEeqnarray*}
for large enough $n$. The notation $\Pbb_{\prod P}$ indicates that the probability is evaluated with respect to the i.i.d. distribution $\prod_k P$. Now, substituting (\ref{step1}) and (\ref{step2}) into (\ref{expect2}), we have
\begin{equation}
 \label{bigstep}
 \Ebb_\pi\bigg[\min_{z}\sum_{x^n\in B}d(x_i,z)\bigg] \geq k\cdot(\min_z \Ebb[d(X,z)]-\delta_2(\eps)).
\end{equation}
Upon substituting (\ref{bigstep}) into (\ref{expect1}), we conclude the proof by noting that
\begin{IEEEeqnarray*}{rCl}
 \frac1n\sum_{i=1}^n\sum_{P\in \Pcal^n_\eps}\sum_{B\in \Bcal_P}c_P\cdot k(n) &=& \Pbb[X^n\in T_\eps^n(P_X)]\\
 &\geq& 1-\eps
\end{IEEEeqnarray*}
for large enough $n$.
\end{proof}

\section{Conclusion}

If an eavesdropper is trying to reconstruct an information sequence in the Shannon cipher system, we have shown that even small amounts of secret key enable the cipher to cause maximal distortion in the eavesdropper's estimate.  Any positive rate of secret key will suffice.  However, the {\em rate} of secret key, implying exponential growth in the number of secret key assignments, is not even the right way to discuss the theoretical limits.  Corollary \ref{finite_cor} shows that the proper question to address is the tradeoff between secret key size and guaranteed distortion, irrespective of the transmission length.

\section{Acknowledgements}

This work was supported by the National Science Foundation (NSF) through the grant CCF-1116013 and by the Defense Advanced Research Projects Agency (DARPA) through the award HR0011-07-1-0002.

\end{document}